%% file: main.tex
\documentclass[11pt,letterpaper]{article}
\usepackage[margin=1in]{geometry}
\usepackage{amsmath,amsthm,amsfonts,amssymb,amscd}
\usepackage{algorithm}
\usepackage[noend]{algpseudocode}
\usepackage{lastpage}
\usepackage{enumerate}
\usepackage{fancyhdr}
\usepackage{mathrsfs}
\usepackage{xcolor}
\usepackage{graphicx}
\usepackage{listings}
\usepackage[colorlinks=true, urlcolor=blue, linkcolor=blue, citecolor=magenta]{hyperref}
\usepackage{cite}
\usepackage{multirow}
\usepackage{comment}
\usepackage{indentfirst}
\usepackage{todonotes}
\usepackage{cleveref}
\usepackage{thm-restate}
\usepackage{enumitem}

\newtheorem{definition}{Definition}

\newtheorem{lemma}{Lemma}

\allowdisplaybreaks

\newcommand{\probms}{\textsc{Makespan}}
\newcommand{\probsc}{\textsc{Santa Claus}}
\newcommand{\probsp}{\textsc{Set Packing}}
\newcommand{\probcsp}{\textsc{CSP}}

\title{Scheduling Problems with Constrained Rejections}

\author{Sami Davies\thanks{Email: {\tt samidavies@berkeley.edu}. Department of EECS, UC Berkeley, and RelationalAI.}
 \and 
Venkatesan Guruswami\thanks{Email: {\tt venkatg@berkeley.edu}. Simons Institute for the Theory of Computing, and Departments of EECS and Mathematics, UC Berkeley. Research supported in part by NSF grants CCF-2228287 and CCF-2211972 and a Simons Investigator award.}
\and
Xuandi Ren\thanks{Email: {\tt xuandi\_ren@berkeley.edu}. Department of EECS, UC Berkeley. Research supported in part by NSF CCF-2228287.}
}

\begin{document}
\maketitle

\abstract{
We study bicriteria versions of Makespan Minimization on Unrelated Machines and Santa Claus by allowing a constrained number of rejections.  Given an instance of Makespan Minimization on Unrelated Machines where the optimal makespan for scheduling $n$ jobs on $m$ unrelated machines is $T$, (Feige and Vondrák, 2006) gave an algorithm that schedules a $(1-1/e+10^{-180})$ fraction of jobs in time $T$. We show the ratio can be improved to $0.6533>1-1/e+0.02$ if we allow makespan $3T/2$. To the best our knowledge, this is the first result examining the tradeoff between makespan and the fraction of scheduled jobs when the makespan is not $T$ or $2T$.

For the Santa Claus problem (the ``Max-Min'' version of Makespan Minimization), the analogous bicriteria objective was studied by (Golovin, 2005), who gave an algorithm providing an allocation so a $(1-1/k)$ fraction of agents receive value at least $T/k$, for any $k \in \mathbb{Z}^+$ and $T$ being the optimal minimum value every agent can receive. We provide the first hardness result by showing there are constants $\delta,\varepsilon>0$ such that it is NP-hard to find an allocation where a $(1-\delta)$ fraction of agents receive value at least $(1-\varepsilon) T$. To prove this hardness result, we introduce a bicriteria version of Set Packing, which may be of independent interest, and prove some algorithmic and hardness results for it. Overall, we believe these bicriteria scheduling problems warrant further study as they provide an interesting lens to understand how robust the difficulty of the original optimization goal might be.}

\input{contents/intro}

\input{contents/alg_ms}
\input{contents/hardness_sc}
\input{contents/alg_sp}

\input{contents/hardness_sp}

\input{contents/conclusion}

\bibliographystyle{alpha}
\bibliography{main}

\appendix

\end{document}

%% file: contents/intro.tex
\section{Introduction}
Two of the most central problems in scheduling theory are Makespan Minimization on Unrelated Machines and its dual, the Santa Claus problem (also called Max-Min Fair Allocation). We refer to the two problems as \probms{} and \probsc{}, in short. 
In \probms{}, a set of jobs $J$ are available to schedule on a set of machines $M$, where the processing time of job $j \in J$ on machine $i \in M$ is $p_{i,j}$. We let $p_{i,j} = \infty$ if job $j$ cannot be processed on machine $i$. The jobs must be assigned non-preemptively, and the goal is to minimize the makespan, i.e., the time the last machine finishes processing.
In \probsc{}, a set of items $I$ are available to assign to a set of agents $A$, where the value agent $i \in A$ receives from item $j \in I$ is $p_{i,j} \geq 0$. Each item can be allocated to at most one agent, and the total value agent $i$ receives is the sum of the values from individual items. The goal is to find an allocation of items to agents so that the minimum total value of any agent is maximized. 

 For \probms{}, the seminal work of Lenstra, Shmoys, and Tardos gave a factor 2 approximation algorithm and
showed that 
the problem is NP-hard to approximate  within a factor better than 3/2 \cite{DBLP:journals/mp/LenstraST90} (essentially it is NP-hard to tell if the makespan is $3$ or $2$ for an instance with nonnegative, integral processing times). Both the hardness and approximation results are still the best known for the general settings.
It is NP-hard to approximate \probsc{} within a factor better than 2 \cite{bezakova2005allocating} (similarly to \probms{}, it is NP-hard to tell if the value is $2$ or $1$ for an instance with nonnegative, integral values). 
The best approximation algorithm for \probsc{} has approximation factor 
$\tilde{O}(m^{\varepsilon})$
and running time $m^{O(1/\varepsilon)}$, for $\varepsilon = \Omega( \log \log m / \log m)$  
\cite{DBLP:conf/focs/ChakrabartyCK09}, where $m$ is the number of items.

For both problems in their most general versions, progress on the hardness and algorithmic fronts has been stuck.  On the hardness side, e.g., for \probsc{}, all natural approaches via local gadgets break down when one can allocate half the intended value to agents. Specifically, \cite{bateni2009maxmin} showed the local gadgets used in all previous reductions can be modified to have \textit{degree two}, i.e., each item is shared by at most two agents. However, \cite{DBLP:conf/focs/ChakrabartyCK09} gave a $(2+\varepsilon)$-approximation algorithm for this special case, ruling out the hope to get better hardness factor using only such gadgets. 
On the algorithmic front, the lack of progress is largely due to the absence of good convex programs for either problem. For instance, the \emph{configuration linear program} has an integrality gap of 2 for \probms{} \cite{verschae2014configuration} and $\Omega(\sqrt{n})$ for \probsc{}, for $n$ the number of agents\footnote{Additionally, there is evidence (but not proof) that even systematic strengthenings of LPs---specifically the Sherali-Adams hierarchy---for \probsc{} may not be strong enough to obtain constant factor approximations in polynomial time \cite{bamas2024lift}}. 

Given the difficulty in improving algorithms or hardness in the general setting, much work over the past several decades has focused on breaking these barriers in special cases. 
The notable cases for \probsc{} include 
the restricted setting  \cite{feige2008allocations, annamalai2017combinatorial, cheng2018restricted, cheng2019restricted,  davies2020tale, haxell2023improved}, unrelated graph balancing \cite{bateni2009maxmin, DBLP:conf/focs/ChakrabartyCK09, verschae2014configuration}, and Max-Min Degree Arborescence \cite{bateni2009maxmin, bamas2023better, bamas2024lift}. For \probms{}, some of the important cases that have been studied include the restricted setting \cite{svensson2011santa, huang2015combinatorial, jansen2017configuration}, graph balancing \cite{EbenlendrKS14, verschae2014configuration,  DBLP:conf/icalp/JansenR19}, and in the fixed-parameter tractable regime \cite{mnich2015scheduling}.
These special settings are valuable because they help the community identify which aspects of the problems really contribute to their difficulty.
For example, one such insight was suggested by Verschae and Wiese \cite{verschae2014configuration} for \probms{}: the configuration linear program has integrality gap $1.833<2$ in the restricted setting, but has gap 2 for unrelated graph balancing, thus hinting that perhaps the difficulty in the problem is present in unrelated graph balancing but not in the restricted setting\footnote{In unrelated graph balancing, every job $j$ has at most two machines where $p_{i,j} < \infty$ but these two processing times may be different, while in the restricted setting, $p_{i,j} \in \{p_j,\infty\}$.}. On the other hand, the unrelated graph balancing setting is much easier to handle for \probsc{}, as there exist several constant factor approximations \cite{bateni2009maxmin, DBLP:conf/focs/ChakrabartyCK09}, indicating that perhaps the difficulties for both problems are not even the same. 

In this paper, we study a different type of relaxation for these problems.
For \probms{} we allow a $1-\alpha$  fraction of jobs to be rejected (i.e., not scheduled on any machine), for \probsc{} we allow a $1-\alpha$ fraction of agents to be rejected (i.e., receive 0 value). 
Then we study how this affects the approximation factor $\beta$ for makespan/value guarantee. 
While a standard approximation algorithm fixes $\alpha=1$ and optimizes $\beta$, we study the trade-off between the two parameters.

This relaxation of \probms{} is the unweighted case of the Maximum General Assignment Problem \cite{FGMS06}, for which Feige and Vondrák \cite{FV06} gave an algorithm that achieved $\alpha=1-1/e+10^{-180}$ and $\beta =1$. On the other hand, Nutov, Beniaminy and Yuster \cite{NBY06} showed the NP-hardness of $\alpha=1-\varepsilon$ and $\beta=1$ for some $\varepsilon>0$. We remark that \cite{FV06}'s improvement over $1-1/e$ relies on a delicate hierarchy of parameters that has to be carefully balanced, which makes it difficult to turn the $10^{-180}$ into a larger constant. In this work, we give an algorithm that reaches $\alpha> 1-1/e+0.02$ and $\beta=3/2$.

\begin{restatable}{theorem}{algms}
\label{thm:bicrit-MS}
For the \probms{} problem, there is a polynomial-time randomized algorithm, which given as input an instance where the optimal makespan for scheduling all $n$ jobs on the $m$ machines is $T$, schedules in expectation $\frac{6e-5}{6e+1}\cdot n>0.6533 \cdot n >(1-1/e+0.02)\cdot n$ many jobs within makespan $\frac{3}{2}\cdot T$.
\end{restatable}

We highlight that the optimal value $T$ we compare against is the optimal makespan for scheduling \emph{all} jobs, not the (potentially smaller) optimal makespan of scheduling a subset. We also remark that scheduling all jobs in makespan less than $3T/2$ is NP-hard \cite{DBLP:journals/mp/LenstraST90}, so when initiating the study of the trade-off between $\alpha$ and $1<\beta<2$, it is natural to consider $\beta=3/2$.

Allowing some jobs to be rejected is not only a natural variant, but we believe it still captures much of the difficulty in the original problems. We conjecture there are some ``transition points'' on the curve between $\alpha$ and $\beta$, where when one wishes to schedule a slightly larger fraction of jobs (say from $\alpha$ to $\alpha+\delta$, for a small constant $\delta>0$), a large jump in the approximation factor for the makespan, $\beta$, is needed.  
One regime where it seems difficult to smooth the curve is when $\beta=1$; it is not clear whether increasing the approximation factor on the makespan slightly allows us to schedule a larger fraction of jobs.
Similarly, points near $(\alpha,\beta)=(1,2)$ are interesting because it seems challenging to design an algorithm that rejects a small constant $\varepsilon$ (say $\varepsilon < 1/10$) fraction of jobs and schedules the rest in makespan $(2-\varepsilon)\cdot T$.

 \probsc{} with constrained rejections has been studied algorithmically by Golovin \cite{golovin2005max}, who gave an algorithm which for any $k \in \mathbb{Z}^+$, finds an allocation for \probsc{} where at least a $(1-1/k)$ fraction of agents obtain value at least $1/k\cdot T$. In this work, we show the hardness of this problem, although in the high-value regime.

\begin{restatable}{theorem}{hardscc}
\label{thm:hardness_sc_2}
    For the \probsc{} problem, there are universal constants $\delta,\varepsilon>0$, such that given the existence of an assignment where all $n$ agents receive total value at least $T$, it's NP-hard to find an assignment where $(1-\delta)\cdot n$ agents receive total value $(1-\varepsilon)\cdot T$.
\end{restatable}

To prove \Cref{thm:hardness_sc_2}, we introduce  bicriteria \probsp{} as an intermediate problem (see \Cref{def: SP}). In the traditional \probsp{} problem, we are given as input a collection $\mathcal{S}$ of sets over a universe $U$ of elements, and the goal is to find the maximum number of disjoint sets from $\mathcal{S}$. 
In bicriteria \probsp{}, instead of being forced to choose a whole set from $\mathcal S$ as in the original formulation, we relax the problem to finding large disjoint subsets of the sets in $\mathcal S$.

While we motivated bicriteria \probms{} and \probsc{} from the long-standing gaps between hardness and approximations, bicriteria \probsp{} is itself an interesting combinatorial optimization problem due to the tight hardness results for \probsp{} (see \Cref{sec: rel-work}). In other words, we know \probsp{} is a hard problem, but our understanding of how rigid the hardness is to relaxing some of the constraints is limited. In this vein, we obtain the following algorithmic and hardness results.
\begin{restatable}{theorem}{algsp}
\label{thm:alg_sp}
    For the \probsp{} problem, for any $\delta>0$, there exists $\varepsilon>0$ and a polynomial-time randomized algorithm, which given as input a collection $\mathcal S$ of $n$ sets over a universe $U$ and the existence of $m$ sets in $\mathcal S$ that form a partition of $U$, finds $m'=(1-\delta)\cdot m$ many sets $S_1,\ldots,S_{m'}$ together with their subsets $A_1 \subseteq S_1,\ldots,A_{m'} \subseteq S_{m'}$, such that for every $i \in [m']$, $|A_i| \ge \varepsilon\cdot |S_i|$, and the sets $A_1,\ldots,A_{m'}$ are disjoint.
\end{restatable}

\begin{restatable}{theorem}{hardsp}
\label{thm:hardness_sp_2}
    For the \probsp{} problem, there are universal constants $\varepsilon>0, \alpha>\beta>0$, such that given as input a collection $\mathcal S$ of $n$ sets over a universe $U$ and the existence of $m=\alpha \cdot n$ sets in $\mathcal S$ that form a partition of $U$, it is NP-hard to find $m'=\beta \cdot n$ many sets $S_1,\ldots,S_{m'}$ in $\mathcal S$ together with their subsets $A_1\subseteq S_1,\ldots,A_{m'} \subseteq S_{m'}$, such that for every $i \in [m'], |A_i| \ge (1-\varepsilon) \cdot |S_i|$, and the sets $A_1,\ldots,A_{m'}$ are disjoint.
\end{restatable}

\begin{table}[htbp]
\centering
\renewcommand{\arraystretch}{1.1}
\caption{A Table Summarizing Bicriteria Results}
\begin{tabular}{c|c|c|c|c}
 &{\small  Algorithm }& {\small Reference} & {\small Hardness } & {\small Reference }\\ \hline
\multirow{3}{*}{\shortstack{{\small$(\alpha,\beta)$-\probms{}}\\ {\small(\Cref{def:ms})}}}
  & {\footnotesize$(1-1/e,1)$ }          & {\footnotesize\cite{FGMS06}  }  & \multirow{3}{*}{{\small$(1-\varepsilon,1)$}} & \multirow{3}{*}{{\small\cite{NBY06}}} \\
  \cline{2-3}
  & {\footnotesize$(1-1/e+10^{-180},1)$} & {\footnotesize\cite{FV06}  }   &                                       &                          \\
  \cline{2-3}
  & {\footnotesize$(1-1/e+0.02,1.5)$}    & {\footnotesize[This work]} &                                       &                          \\ \hline

\multirow{3}{*}{\shortstack{{\small$(\alpha,\beta)$-\probsc{}}\\ {\small(\Cref{def:sc})}}}
  & \multirow{3}{*}{{\small$(1-1/k,1/k)$}}
  & \multirow{3}{*}{{\small\cite{golovin2005max}}}
  & \multirow{3}{*}{{\small$(1-\delta,1-\varepsilon)$}}
  & \multirow{3}{*}{{\small[This work]}} \\
  & & & & \\
  & & & & \\ \hline

\multirow{3}{*}{\shortstack{{\small$(\alpha,\beta)$-\probsp{}}\\ {\small(\Cref{def: SP})}}}
  & \multirow{3}{*}{{\small$(1-\delta,\varepsilon)$}}
  & \multirow{3}{*}{{\small[This work]}}
  & \multirow{3}{*}{{\small$(1-\delta,1-\varepsilon)$}}
  & \multirow{3}{*}{{\small[This work]}} \\
  & & & & \\
  & & & & \\ \hline
\end{tabular}
\label{table:1}
\end{table}

 We note that in bicriteria \probsp{}, we assume the existence of $m$ sets partitioning $U$, while in general \probsp{}, this might not be the case.

We give the algorithm for bicriteria \probms{} (\Cref{thm:bicrit-MS}) in \Cref{sec:alg_ms}. We show the hardness for bicriteria \probsc{} (\Cref{thm:hardness_sc_2}) in \Cref{sec:hardness_sc}. The algorithm and the hardness for bicriteria \probsp{} (\Cref{thm:alg_sp} and \Cref{thm:hardness_sp_2}) are in \Cref{sec:alg_sp} and \Cref{sec:hardness_sp}. 

\subsection{Related works}\label{sec: rel-work}
Some of the broader categories for bicriteria objectives in scheduling
include 
(1)  minimizing a linear combination of different objectives \cite{angel2001fptas, hoogeveen2003preemptive, mnich2015scheduling, gupta2023bicriteria}, 
(2) minimizing one objective while constraining the others \cite{golovin2005max, shabtay2012two,  shabtay2013survey}, and
(3) obtaining a Pareto optimal with respect to different objectives \cite{papadimitriou2000approximability, angel2001fptas, shabtay2012two, li2023bicriteria}.
Note that Pareto optimal solutions (category (3)) also satisfy the criteria of categories (1) and (2).
Our problems of \probms{} and \probsc{} with a constrained number of rejected jobs or agents are most similar to those into category (2).  
Scheduling with rejections are problems in bicriteria scheduling where one of the objectives is a function of the set of rejected jobs \cite{bartal2000multiprocessor, hoogeveen2003preemptive, angel2001fptas, TkindtBillaut, shabtay2013survey, pei2020new, liu2020approximation,  liu2021approximation}. 
In the works studying the trade-off between makespan and the rejection penalty using the Pareto optimality criteria, results either only hold for 2 machines \cite{shabtay2012two}, or are only polynomial time for a constant number of machines \cite{angel2001fptas, shabtay2013survey}.

For the classic problem of minimizing the makespan on unrelated machines, it is a long-standing open problem in scheduling theory to develop a polynomial time algorithm that beats the 2-approximation \cite{DBLP:journals/mp/LenstraST90}. 
Several lines of work explore additional assumptions on the input that make it easier to break the boundary of 2, including in the FPT setting \cite{horowitz1976exact, jansen1999improved, mnich2015scheduling} and in graph balancing 
\cite{EbenlendrKS14, huang2015combinatorial, DBLP:conf/icalp/JansenR19}. 
Algorithms in the restricted assignment setting ($p_{i,j} \in \{\infty, p_j\}$)   estimate the optimal makespan within factors strictly less than 2 but do not necessarily terminate in polynomial time \cite{svensson2011santa, jansen2017configuration}. 

\probsc{} was first introduced in the job-machine scheduling setting \cite{woeginger1997polynomial, epstein1999approximation}. 
The configuration LP was introduced for \probsc{} by Bansal and Sviridenko  \cite{bansal2006santa}.
Asadpour and Saberi prove that the CLP has integrality gap $\Omega(\sqrt{n})$ 
and round the solution to the CLP to obtain a $\tilde{O}(\sqrt{n})$ approximation algorithm \cite{asadpour2007approximation}, before the improvement to a $\tilde{O}(m^{\varepsilon})$-approximation algorithm with running time $m^{O(1/\varepsilon)}$, for $\varepsilon = \Omega( \log \log m / \log m)$, by Chakrabarty, Chuzhoy, and Khanna 
\cite{DBLP:conf/focs/ChakrabartyCK09}. 
Recently, Bamas et al. showed that for any $\alpha \ge 2,\varepsilon>0$, a polynomial-time $(2-1/\alpha)$-approximation for \probms{} leads to a polynomial-time $(\alpha+\varepsilon)$ approximation for \probsc{}, and under the two-value restriction, they are actually equivalent \cite{BLMRS24}.

For the bicriteria \probms{} problem, Chekuri and Khanna \cite{chekuri2005polynomial} observed that the algorithm of Lenstra, Shmoys, and Tardos already schedules at least a $1/2$ of the jobs in makespan $T$. 
In the case when $\beta=1$, bicriteria \probms{} is a special case of the Maximum General Assignment Problem.   
Previous work on the Maximum General Assignment Problem by
Fleischer, Goemans, Mirrokni, and Sviridenko \cite{FGMS06} proves one can schedule a $1-1/e$ fraction of jobs in makespan $T$. 
Feige and Vondrák \cite{FV06} improved the fraction to $1-1/e+10^{-180}$. The best NP-inapproximability ratio of the Maximum General Assignment Problem is $10/11$ by Chakrabarty and Goel \cite{CG10}. However, in their construction, the jobs are weighted, which makes it incapable to fit into our setting. It was proven by Nutov, Beniaminy, and Yuster \cite{NBY06} that the unit weight case (a.k.a. maximizing the number of jobs scheduled) is also APX-hard.

For bicriteria \probsc{}, Golovin proposed an algorithm which for any $k \in \mathbb{Z}^+$, finds an allocation for \probsc{} where  a $(1-1/k)$ fraction of agents obtain value $1/k\cdot T$ \cite{golovin2005max}.

For the \probsp{} problem where we are given as input a collection $\mathcal{S}$ of sets over a universe $U$, the goal is to find the maximum number of disjoint sets from $\mathcal{S}$. By the simple correspondence between $\textsc{Independent Set}$ and \probsp{}, there is a polynomial time $O\left(\frac{|\mathcal S|(\log \log |S|)^2}{(\log |S|)^3}\right)$-approximation algorithm \cite{Fei04} for it, and it cannot be approximated within a factor less than $\Omega(|\mathcal{S}|^{1-\varepsilon})$ for any $\varepsilon>0$ unless P=NP \cite{haastad1999clique, zuckerman2006linear}. In terms of the universe size $|U|$, Halldórsson, Kratochvíl and Telle gave an algorithm that admits an $O(\sqrt{|U|})$ approximation \cite{halldorsson2000independent}, which matches the $\textsc{Independent Set}$-based hardness $\Omega(|U|^{1/2-\varepsilon})$ \cite{haastad1999clique, zuckerman2006linear}. The special case where each set has a bounded size $k$ (known as $k$-\probsp{}) is also widely studied \cite{HS89,Hal95,Ber00,LOSZ20,Neu21,Neu23}. The state-of-the-art approximation ratios are $O(k)$ \cite{Cyg13,LSV13}, while the hardness was recently improved from the long-standing $\Omega(k/\log k)$ \cite{HSS06} to $\Omega(k)$ \cite{LST24}.

\subsection{Preliminaries}


We begin with preliminaries for \probms{}.
It is helpful to consider the bipartite graph $G=(V,E)$ where $V=M \dot\cup J$ and $E$ is the set of (machine, job) pairs $(i,j)$ where $p_{i,j} <\infty$. A \textit{schedule} $\mathcal O$ is simply a subset of $E$. $T$ denotes the optimal makespan.
In our algorithms for $(\alpha,\beta)$-bicriteria \probms{}, we partition $E$ into disjoint sets of \emph{large} and \emph{small} edges, $E_L$ and $E_S$, according to the processing times
$
    E_L  =\{(i,j) \mid i \in M, j \in J, \frac{T}{2} < p_{i,j}$$ \le T\},$ and $
    E_S  =\{(i,j) \mid i \in M,j \in J, 0 \le $$ p_{i,j} \leq \frac{T}{2}\},$
and define $G_L$ and $G_S$ to be the bipartite graphs induced by $E_L$ and $E_S$, respectively. We let $m^{\star}$ denote the size of a maximum matching in $G_L$. 
When we input an instance of \probms{} into an algorithm for $(\alpha,\beta)$-bicriteria \probms, we denote the input with the shorthand $\Gamma=(J,M,T,E')$, for $J$ the jobs, $M$ the machines, $T$ the optimal makespan, and $E' \subseteq E$ the edges along which we make assignments.

\begin{definition}[\probms{} Problem]
\label{def:ms}
    In the \probms{} problem, we are given as input a set of $m$ machines $M$, a set of $n$ jobs $J$, and a table $\{p_{i,j}\}_{i \in M,j \in J}$, where $p_{i,j}$ is the processing time of job $j$ on machine $i$. We write $p_{i,j}=\infty$ if job $j$ cannot be scheduled on machine $i$. The goal is to assign jobs to machines, such that the maximum total processing time on any machine is minimized. 

     For $0\le \alpha \le 1 \le \beta \le 2$, $(\alpha,\beta)$\emph{-bicriteria} \probms{} is the task where given the optimal makespan $T$, the goal is to schedule a subset of $\alpha\cdot n$ jobs on $M$ within makespan $\beta \cdot T$.
\end{definition}

The definitions for \probsc{} and \probsp{} and their bicriteria variants are as follows.  Other notations are self-contained in those relevant sections.

\begin{definition}[\probsc{} Problem]
\label{def:sc}
    In the \probsc{} problem, we are given as input a set of $n$ agents $A$, a set of $m$ items $I$, and a table $\{p_{i,j}\}_{i \in A, j \in J}$ where $p_{i,j}$ is the value agent $i$ can get on receiving item $j$. The goal is to assign items to agents to maximize the minimum total value of any agent.

    For $0 \le \alpha,\beta \le 1$, $(\alpha,\beta)$\emph{-bicriteria} \probsc{} is the task where given the optimum value $T$, the goal is to
    allocate items to agents so that at 
     least $\alpha \cdot n$ agents have total value at least $\beta\cdot T$.
\end{definition}

\begin{definition}[\probsp{} Problem]\label{def: SP}
    In the \probsp{} problem, we are given as input a collection of sets $\mathcal S$ over a universe $U$, and the goal is to find the maximum number of disjoint sets in $\mathcal S$.

    For $0 \le \alpha,\beta \le 1$, $(\alpha,\beta)$\emph{-bicriteria} \probsp{} is the task where given the guarantee that there are $m$ sets in $\mathcal S$ partitioning $U$, the goal is to find $m'=\alpha \cdot m$ many sets $S_1,\ldots,S_{m'}$ in $\mathcal S$ together with their subsets $A_1\subseteq S_1,\ldots,A_{m'} \subseteq S_{m'}$, such that for every $i \in [m'], |A_i| \ge \beta \cdot |S_i|$, and the sets $A_1,\ldots,A_{m'}$ are disjoint.
\end{definition}

%% file: contents/alg_ms.tex
\section{An Algorithm for $(0.6533, 1.5)$-bicriteria \probms{}}
\label{sec:alg_ms}

In this section, we prove \Cref{thm:bicrit-MS}. 

\algms*

Our algorithm has three sub-procedures:
\begin{enumerate}
    \item a deterministic \Cref{alg:1}, which simply computes a maximum matching on $G_L$ and schedules these $m^{\star}$ jobs within makespan $T$;
    \item a randomized \Cref{alg:2}, which given the promise that all jobs can fit into makespan $T$, schedules a $(1-\frac{1}{e})$ fraction of jobs in expectation within makespan $T$;
    \item a deterministic \Cref{alg:3}, which schedules $\frac{1}{6}(n-m^{\star})$ many jobs within makespan $\frac{T}{2}$, using only edges in $E_S$. 
\end{enumerate}
While each sub-procedure on its own is simple (especially \Cref{alg:1} and \Cref{alg:2}),  the salient point is that they can be combined in a way that is stronger than any of them individually. 

We describe each algorithm, then show how to combine them in \Cref{sec:alg_combined}.

\subsection{\Cref{alg:1}: The Matching Algorithm on Large Jobs}
\label{sec:alg_1}

 \Cref{alg:1} simply finds a matching in $G_L$ and schedules these $m^{\star}$ jobs. Given the result of \Cref{alg:1}, the following \Cref{lem:lb_using_only_es} shows a lower bound on the number of jobs that can be scheduled using only small edges. 

\begin{algorithm}[h!]
\caption{Matching Large Edges}\label{alg:1}
\begin{algorithmic}[1]
  \State \textbf{Input}: A \probms{} instance $\Gamma=(J,M,T,E_L)$.
  \State Find a maximum matching $\mathcal{O}^{\star}$ on the  bipartite graph $G_L=(V=M \dot\cup J, E=E_L)$.
  \State Output $\mathcal{O}^{\star}$.
\end{algorithmic}
\end{algorithm}

\begin{restatable}{lemma}{algmatch}
\label{lem:lb_using_only_es}
    It is possible to schedule at least $n-m^{\star}$ many jobs within makespan $T$ using only edges in $E_S$, where $m^{\star}$ is the matching size output by \Cref{alg:1}. 
\end{restatable}

\begin{proof}
    Recall the guarantee that it's possible to schedule all $n$ jobs of $J$ on $M$ within makespan $T$. Let $\mathcal{O} \subseteq E$ be such a schedule. For every machine $i \in M$, there is at most one job $j$ scheduled on it in $\mathcal{O}$ with $(i,j) \in E_L$, as otherwise the makespan will exceed $T$. Therefore, $\mathcal{O}  \cap E_L$ forms a matching on $G_L$.  Since the matching output by \Cref{alg:1} is a maximum matching,
    we have $m^{\star} \ge |\mathcal{O}  \cap E_L|$. 
    
    By simply throwing out edges in $\mathcal{O}  \cap E_L$ from the schedule $\mathcal{O} $, we prove the existence of a schedule using only $E_S$, as 
    \[
        |\mathcal{O} \cap E_S|= |\mathcal{O} |-|\mathcal{O}  \cap E_L| \ge |\mathcal{O} |-m^{\star} =n-
        m^{\star}. 
         \qedhere
    \]
\end{proof}

\subsection{\Cref{alg:2}: A Randomized LP-rounding Algorithm}

We define the \emph{Configuration LP} for \probms{}, which we sometimes abbreviate as the CLP. For machine $i \in M$, let $
\mathcal{C}_i$ be the subsets of jobs that have combined processing time at most $T$ on machine $i$, i.e. $\mathcal{C}_i= \{C \subseteq J \mid \sum_{j \in C}p_{i,j} \leq T\}$. Let $y_{i,C}$ be the variable indicating whether configuration $C$ is scheduled on machine $i.$ Then the configuration LP is the following:
\begin{equation}\label{eq:lp}
\begin{aligned}
    \sum_{C \in \mathcal{C}_i}y_{i,C} &= 1 \hspace{1mm}  \forall i \in M, \quad
     \sum_{i \in M}\sum_{C \in \mathcal{C}_i : j \in C}y_{i,C} = 1 \hspace{1mm}  \forall j \in J, \quad
     y_{i,C} \geq 0 \hspace{1mm}  \forall i \in M, \forall C \in \mathcal{C}_i.
\end{aligned}
\end{equation}

If all jobs can be scheduled within makespan $T$ in an integral solution, then there is a feasible solution to the CLP. 
Further, a feasible solution to the CLP can be found to any degree of accuracy in polynomial time~\cite{bansal2006santa, svensson2011santa}. 

Given a solution to the CLP, the rounding algorithm simply lets each machine sample a configuration $C$ in $\mathcal{C}_i$ with probability $y_{i,C}$, then schedule all unscheduled jobs that it samples. See the pseudo-code in \Cref{alg:2}.  


\begin{algorithm}[H]
\caption{Randomized Rounding of the CLP}\label{alg:2}
\begin{algorithmic}[1]
  \State \textbf{Input}: A \probms{} instance $\Gamma=(J,M,T,E)$.
  \State Solve the CLP defined in \Cref{eq:lp}.
  \State Initialize $\mathcal{O}^{\star}=\emptyset$.
  \For{every machine $i \in M$}
    \State Sample a configuration $C \in \mathcal C_i$ with probability $y_{i,C}$.
    \For{every job $j$ in the sampled configuration}
    \State Add $(j,i)$ to $\mathcal{O}^{\star}$ if $j$ is not yet scheduled in $\mathcal{O}^{\star}$.
    \EndFor
  \EndFor
  \State Output $\mathcal{O}^{\star}$.
\end{algorithmic}
\end{algorithm}

\begin{restatable}{lemma}{algrandom}
\label{lem:alg_2}
    Given the promise that all jobs can be integrally scheduled within makespan $T$, \Cref{alg:2} schedules at least a  $(1-\frac{1}{e})$ fraction of the jobs in expectation within makespan $T$.
\end{restatable}

\begin{proof}
A job is not scheduled only if it is not sampled by any machine. Thus,
\begin{align*}
    \mathbb{E}[\text{number of scheduled jobs}] &= \sum_{j \in J} \mathbb{P}[j \text{ sampled}] \\
    &= 
    \sum_{j \in J} \Big (1-\prod_{i \in M}\mathbb{P}[j \text{ not  sampled on machine }i] \Big )\\
    & = \sum_{j \in J} \Big (1-\prod_{i \in M}\Big ( 1-\sum_{C \in \mathcal{C}_i: j \in C} y_{i,C}
    \Big ) \Big )\\
    &\geq \sum_{j \in J} \Big (1-e^{-\sum_{i \in M}\sum_{C \in \mathcal{C}_i: j \in C} y_{i,C}}
    \Big ) \\
    &= \sum_{j \in J} \left(1-\frac{1}{e}
    \right) = \left(1-\frac{1}{e}\right)\cdot n.  
     \qedhere
\end{align*}

\end{proof}
We note that this result was known from \cite{FGMS06}.

\subsection{\Cref{alg:3}: A Greedy Algorithm using only Small Jobs}

Our \Cref{alg:3} deterministically schedules $\frac{1}{6}(n-m^{\star})$ many jobs within makespan $\frac{T}{2}$ using only edges in $E_S$, where $m^{\star}$ is the maximum matching size on $G_L$ as computed by \Cref{alg:1}. \Cref{alg:3} first computes a schedule of makespan $T$ by greedily picking the job with minimum processing time\footnote{This is just a specific instantiation of the  classic \emph{List Scheduling} algorithm \cite{graham1969bounds}.}, then truncates each machine's schedule up to time $\frac{T}{2}$. 

\begin{algorithm}[h!]
\caption{Deterministic Greedy Algorithm}\label{alg:3}
\begin{algorithmic}[1]
    \State \textbf{Input}: A \probms{} instance $\Gamma=(J,M,T,E_S)$.
    \State Initialize $c_i \gets 0,d_i \gets 0$ for all $i \in M$.
    \State Set up an empty list $\ell_i$ for all $i \in M$.
    \State Initialize $R$ to be $J$, the set of all jobs.
    \While{$R \neq \emptyset$}
      \State Let $p_{i^{\star},j^{\star}}$ be the minimum $p_{i,j}$ such that $j \in R$ and $c_i + p_{i,j} \le T$.
      \If{Such a $(i^{\star},j^{\star})$ exists}
        \State Append $j^\star$ to the list $\ell_{i^\star}$.
        \State $c_{i^\star} \gets c_{i^\star} + p_{i^\star,j^\star}$.
        \State Remove $j^{\star}$ from $R$.
      \Else
        \State \textbf{break}
      \EndIf
    \EndWhile
    \State Initialize $\mathcal{O}^{\star} = \emptyset$.
    \For{every machine $i \in M$}
      \State Sort the list $\ell_i$ according to increasing $p_{i,j}$.
      \For{job $j$ in the list $\ell_i$}
        \If{$d_i + p_{i,j} \le \frac{T}{2}$}
          \State Add $(j,i)$ to $\mathcal{O}^{\star}$.
          \State Update $d_i \gets d_i + p_{i,j}$.
        \Else
          \State \textbf{break}
        \EndIf
      \EndFor
    \EndFor
    \State Output $\mathcal{O}^{\star}$.
\end{algorithmic}
\end{algorithm}

\begin{lemma}\label{lem:ms_det}
    Let $m^{\star}$ be maximum matching size on $G_L$ computed by \Cref{alg:1}. \Cref{alg:3} schedules at least $\frac{1}{6}(n-m^{\star})$ many jobs within makespan $\frac{T}{2}$.
\end{lemma}
We split the proof of \Cref{lem:ms_det} into the following \Cref{lem:alg_3_1} and \Cref{lem:alg_3_2}:

\begin{lemma}\label{lem:alg_3_1}
    At the beginning of Line 13 in \Cref{alg:3}, we have $\sum_{i \in M} |\ell_i| \ge \frac{1}{2}(n-m^{\star})$.
\end{lemma}

To prove \Cref{lem:alg_3_1}, we show in a recursive way that in the \texttt{while} loop in Line 5-12, starting with any available job set $R$ and any vector $c=(c_1,\ldots,c_m)$ of current completion times of all machines, our algorithm adds at least half of the jobs to the lists as compared to any other feasible schedule. This is formalized in the following \Cref{prop:alg_3}.

\begin{restatable}
{proposition}{proppf}
\label{prop:alg_3}
    Let $R\subseteq J$ be an arbitrary set of jobs and $c=(c_1,\ldots,c_m)$, where $ \forall i \in M, 0 \le c_i \le T$, be a vector of the current completion times of the machines. If it's possible to schedule a set $S \subseteq R$ of jobs continuing from $c$ within makespan $T$, then the while loop in \Cref{alg:3}, starting with configuration $(R,c)$, adds at least $\frac{1}{2}|S|$ many jobs to the lists.
\end{restatable}

\begin{proof}
    If the maximum scheduleable set $S=\emptyset$, there is nothing to show. Otherwise, fix a choice of the maximum $S$ and an arbitrary schedule $\mathcal{O}$ of $S$ on the machines continuing from $c$. Let $(i^{\star},j^{\star})$ be the (machine, job) pair picked by \Cref{alg:3}, which has minimum processing time among all available pairs.

    Suppose $j^{\star} \notin S$, i.e., $\mathcal O$ does not schedule $j^{\star}$ anywhere. Then $\mathcal O$ must have scheduled some other job $j' \in S$ on machine $i^{\star}$, since, otherwise, $\mathcal O \cup \{(i^{\star},j^{\star})\}$ is a valid schedule, contradicting the maximality of $S$. Furthermore, $p_{i^{\star},j^{\star}} \le p_{i^{\star},j'}$ as scheduling $j'$ on $i^{\star}$ was a valid option but we chose the smallest available processing time. Thus, $S \setminus \{j'\} \cup \{j^{\star}\}$ is also a maximum scheduleable set with $O \setminus \{(i^{\star},j')\} \cup \{(i^{\star},j^{\star})\}$ being the schedule. 
    
    In the following, we assume without loss of generality that $j^{\star} \in S$. Again, we can without loss of generality assume $\mathcal O$ schedules some job $j'$ on machine $i^{\star}$ ($j'$ could be another job or $j^{\star}$ itself). Given $p_{i^{\star},j^{\star}} \le p_{i^{\star},j'}$, it's possible to schedule the set $S \setminus \{j^{\star}, j'\}$, which is a subset of $R'=R \setminus \{j^{\star}\}$, starting with the completion times $c'=(c_1,\ldots,c_{i^{\star}}+p_{i^{\star},j^{\star}},\ldots,c_m)$.  
    Invoking \Cref{prop:alg_3} on $(R',c')$, we have at least $\frac{1}{2}|S \setminus \{j^{\star}, j'\}|$ many jobs in the lists. Adding the job $j^{\star}$ gives us $\frac{1}{2}|S \setminus \{j^{\star},j'\}|+1 \ge \frac{1}{2}|S|$, which finishes the proof.
\end{proof}

Given \Cref{prop:alg_3}, \Cref{lem:alg_3_1} follows by plugging the set of $n-m^{\star}$ jobs guaranteed by \Cref{lem:lb_using_only_es} into the set $S$ in \Cref{prop:alg_3}, and setting $c=(0,\ldots,0)$.

\begin{lemma}\label{lem:alg_3_2}
    Line 13 $\sim$ Line 20 of \Cref{alg:3} schedules $\frac{1}{3} \sum_{i \in M}|\ell_i|$ jobs within makespan $\frac{T}{2}$.
\end{lemma}

\begin{proof} 
Consider each machine $i$.
    If $|\ell_i|=1$, since all edges are in $E_S$, every $p_{i,j}$ is at most $\frac{T}{2}$ and we can always schedule that unique job on machine $i$.
    
    If $|\ell_i| \ge 2$, then given $\sum_{j \in \ell_i} p_{i,j} \le T$, the smaller $\lfloor \frac{|\ell_i|}{2}\rfloor$ many jobs have total processing time at most $\frac{T}{2}$. If $|\ell_i|=2$, then $\lfloor \frac{|\ell_i|}{2}\rfloor=\frac{|\ell_i|}{2}$. Otherwise for $|\ell_i| \ge 3$, we have $\lfloor \frac{|\ell_i|}{2}\rfloor \ge \frac{|\ell_i|-1}{2} \ge \frac{|\ell_i|}{2}-\frac{|\ell_i|}{6}=\frac{|\ell_i|}{3}$.
\end{proof}

\Cref{lem:ms_det} follows from combining \Cref{lem:alg_3_1} and \Cref{lem:alg_3_2}.

\medskip

We note that \Cref{lem:ms_det} is quite brittle with the makespan of $T/2$. One may hope it is possible to schedule a smaller fraction of jobs in a smaller makespan and have a smooth trade-off here, which would necessitate the processing time used to define the edges in $E_S$ be smaller. However if we partition the edges in $E$ into more subsets, or if we make the threshold processing time between $E_L$ and $E_S$ smaller, we lose the guarantee that at most one edge in $E \setminus E_S$ can be integrally scheduled per machine. Overall, it's not clear how to make this trade-off worth it in the analysis so that a higher fraction of jobs can be scheduled.

\subsection{The Combined Algorithm}
\label{sec:alg_combined}
In this subsection, we show how to combine the 3 algorithms to prove \Cref{thm:bicrit-MS}. The final algorithm outputs the better one of the following two schedules:

\newcounter{schedulings}
\addtocounter{schedulings}{1}

\begin{enumerate}[label=\textbf{(S\arabic{schedulings})}]
    \item \label{sch:1} The schedule output by \Cref{alg:1}.
    \addtocounter{schedulings}{1}
    
    \item \label{sch:2} The schedule obtained by first running \Cref{alg:3}, then \Cref{alg:2} on the remaining jobs.
    \addtocounter{schedulings}{1}
\addtocounter{schedulings}{1}
\end{enumerate}

\begin{proof}[Proof of \Cref{thm:bicrit-MS}]
    \Cref{sch:1} schedules $m^{\star}$ jobs within makespan $T$.
    
    We now analyze \Cref{sch:2}. By \Cref{lem:ms_det}, \Cref{alg:3} schedules at least $\frac{1}{6}(n-m^{\star})$ many jobs within makespan $\frac{T}{2}$. Recall the guarantee that there is a way to schedule all $n$ jobs within makespan $T$, thus it's also possible for the remaining $n-\frac{1}{6}(n-m^{\star})$ many jobs to fit into makespan $T$. By \Cref{lem:alg_2}, \Cref{alg:2} schedules at least a $(1-\frac{1}{e})$ fraction of them in expectation.

\medskip

    Combining \Cref{sch:1} and \Cref{sch:2}, we
schedule\footnote{
A natural question is whether one could use the 2-approximation of Lenstra, Shmoys, and Tardos (or a subroutine of it) to schedule a large set of jobs in makespan $T/2$, then combine this with Algorithm \ref{alg:2}. We found that such a procedure might do better for some $m^{\star}$ values, but it does not help the worst-case.} \[\max\left(m^{\star},
    \frac{1}{6}\left(n-m^{\star}\right)+\bigg(1-\frac{1}{e}\bigg)\bigg(n-\frac{1}{6}\big(n-m^{\star}\big)\bigg)
    \right)\] many jobs within makespan $\frac{3}{2}T$. The minimum is attained when $m^{\star}=\frac{6e-5}{6e+1}\cdot n$. 
\end{proof}

%% file: contents/hardness_sc.tex
\section{Hardness for $(1-\delta,1-\varepsilon)$-bicriteria \probsc{}}
\label{sec:hardness_sc}

In this section, we prove \Cref{thm:hardness_sc_2}.

\hardscc*

We reduce from the following hardness of $(\frac{\beta}{\alpha},1-\varepsilon)$-bicriteria \probsp{}.

\hardsp*

We first prove \Cref{thm:hardness_sc_2} given \Cref{thm:hardness_sp_2}. The proof of \Cref{thm:hardness_sp_2} is in \Cref{sec:hardness_sp}.

\begin{proof}[Proof of \Cref{thm:hardness_sc_2}]
    We set $\delta= \alpha-\beta$, and reduce a \probsp{} instance as in \Cref{thm:hardness_sp_2} to a \probsc{} instance as follows.

    \begin{itemize}
        \item Set up $n$ agents which correspond to the $n$ sets.
        \item Set up $|U|+(1-\alpha)\cdot n$ items, where $|U|$ of them are normal items corresponding to the elements in the universe, and the remaining $(1-\alpha)\cdot n$ items are dummy items.
        \item For an agent $i$ that corresponds to a set $S_i$ and a normal item $j$ that corresponds to an element $j \in U$, we have $p_{i,j} = T/|S_i|$ if $j \in S_i$ and 0 otherwise. 
        For any agent $i$ and any dummy item $j$, we set $p_{i,j}=T$.
    \end{itemize}

    For completeness, we can let the $m=\alpha \cdot n$ agents that correspond to the $m$ non-intersecting sets take their interested items, and distribute a dummy item to each of the remaining $(1-\alpha)\cdot n$ agents.

    For soundness, suppose there are $(1-\delta) \cdot n = (1-\alpha+\beta)\cdot n$ many agents who get happiness at least $(1-\varepsilon) \cdot T$. Among them at most $(1-\alpha)\cdot n$ agents are satisfied by the dummy items, so at least $m'=\beta \cdot n$ agents get happiness $(1-\varepsilon) \cdot T$ by only picking normal items. This gives us a solution to the $(\frac{\beta}{\alpha},1-\varepsilon)$-bicriteria \probsp{}.
\end{proof}

%% file: contents/alg_sp.tex
\section{An Algorithm for $(1-\delta,\varepsilon)$-bicriteria \probsp{}}
\label{sec:alg_sp}

In this section, we prove \Cref{thm:alg_sp}. 

\algsp*


We first introduce the following version of a Chernoff Bound.

\begin{lemma}[Multiplicative Chernoff Bound]\label{lem:chernoff}
    Suppose $x_1,\ldots,x_n$ are independent random variables taking values in $\{0,1\}$. Let $X$ denote their sum and let $\mu =\mathbb E[X]$ denote the sum's expected value. Then for any $\gamma>0$, 
    $\Pr\left[X \ge (1+\gamma) \mu\right] \le \exp{\left(-\frac{\gamma^2 \mu}{2+\gamma}\right)},$
    and for any $0 <\gamma<1$,
    $\Pr\left[X \le (1-\gamma) \mu\right] \le \exp{\left(-\gamma^2 \mu/2\right)}.$
\end{lemma}

 Fix any $\delta>0$, we set  constants $D=40/\delta$, $C=400/\delta^2$, and $\varepsilon=\delta^2/800$.
We then divide the sets into two groups based on whether their sizes exceed $C$. Specifically, let
$
    \mathcal S_{s}  = \{S \in \mathcal S \mid |S| \le C\}$ and $
    \mathcal S_{b}  = \{S \in \mathcal S \mid |S| > C\}.
$
Note that as $\varepsilon<1/C$, for a set in $\mathcal S_s$, picking 1 element suffices to represent an $\varepsilon$ fraction of it. Therefore, we first allocate 1 element to as many sets in $\mathcal S_s$ as possible using a Max-flow algorithm (\Cref{alg:sp_1}), then run an LP-rounding based algorithm (\Cref{alg:sp_2}) on the remaining instance. See the following pseudo-codes for the description of the algorithms. Throughout this section, when we consider some subproblem $(\mathcal S',U')$, we mean one is only allowed to pick sets in $\mathcal S'$ that are \textit{fully contained} in $U'$.

\begin{algorithm}
\caption{A Polynomial-time Randomized Algorithm for $(1-\delta,\varepsilon)$-\probsp{}}
\label{alg:sp_all}
\begin{algorithmic}
\State \textbf{Input}: a \probsp{} instance $\Gamma=(\mathcal S,U)$, where $\mathcal S= \mathcal S_s \dot\cup \mathcal S_b$.
\State Run \Cref{alg:sp_1} on $(\mathcal S_s,U)$ and let $(\mathcal T_s,U_s)$ be the output.
\State Run \Cref{alg:sp_2} on $(\mathcal S_b, U \setminus U_s)$ and let $(\mathcal T_b,U_b)$ be the output.
\State \textbf{Output}: $\mathcal T_s \cup \mathcal T_b$.
\end{algorithmic}
\end{algorithm}

\begin{algorithm}
\caption{An Algorithm for the Small-sized Sets}
\label{alg:sp_1}
\begin{algorithmic}
\State \textbf{Input}: A collection of sets $\mathcal S$ over the universe $U$, where every set has size $\le C$.
\State Build the network $G$ containing a source $s$, a sink $t$, a node for every $S \in \mathcal S$, and a node for every $u \in U$.
\State Add edge $(s,S)$ with capacity 1 to $G$ for every $S \in \mathcal S$.
\State Add edge $(u,t)$ with capacity 1 to $G$ for every $u \in U$.
\State Add edge $(S,u)$ with capacity 1 to $G$ for every $u \in S$.
\State Run a Max-flow algorithm on $G$.
\State Let $\mathcal T$ be the collection of sets $S$ with edge $(s,S)$ saturated.
\State Let $U'$ be the collection of elements $u$ with edge $(u,t)$ saturated.
\State \textbf{Output}: $(\mathcal T,U')$.
\end{algorithmic}
\end{algorithm}

The following lemma shows the optimality of working in two phases as in \Cref{alg:sp_all}. For a \probsp{} instance $\Gamma=(\mathcal S,U)$, we use $\text{OPT}(\Gamma)$ to denote the maximum number of disjoint sets in it.

\begin{lemma}\label{lem:two-phase}
    Let $\Gamma=(\mathcal S,U)$ be a \probsp{} instance. For any $U'\subseteq U$, we have $|U'|+\text{OPT}(\mathcal S,U \setminus U')\ge \text{OPT}(\mathcal S,U)$.
\end{lemma}
\begin{proof}
    Starting with the optimal solution for $(\mathcal S,U)$, removing an element from $U$ only affects at most one set since the picked sets are disjoint. Doing this for every element in $U'$ finishes the proof.
\end{proof}

Let $(\mathcal T_s,U_s)$ be the output of \Cref{alg:sp_1}. Clearly $|\mathcal T_s|=|U_s|$. By \Cref{lem:two-phase}, we have $|\mathcal T_s|+\text{OPT}(\mathcal S,U \setminus U_s)\ge \text{OPT}(\mathcal S,U)$. By the optimality of \Cref{alg:sp_1}, every set in $\mathcal S_s$ must intersect $U_s$, thus $\text{OPT}(\mathcal S,U \setminus U_s)=\text{OPT}(\mathcal S_b,U \setminus U_s)$.
Note that \Cref{alg:sp_2} will solve the following LP:

\begin{equation}\label{eq:lp2}
    \begin{aligned}
        \text{maximize} \quad  \sum_{S} x_S \qquad 
        \text{s.t.} \quad    \sum_{S \ni u} x_S \le 1 \quad  \forall u \in U,
       \quad  x_S \in [0,1]  \quad  \forall S \in \mathcal S.
    \end{aligned}
\end{equation}

\begin{algorithm}
\caption{An Algorithm for the Large-sized Sets}
\label{alg:sp_2}
\begin{algorithmic}
\State \textbf{Input}: A collection of sets $\mathcal S$ over the universe $U$, where every set has size $> C$. Parameter $D$.
\State Solve LP (\ref{eq:lp2}) and let $x$ be an optimal solution.
\State Initialize $\mathcal A = \emptyset$.
\State Add $S$ to $\mathcal A$ with probability $x_S$ for all $S \in \mathcal{S}$.
\State Let $T = \{ u \in U \mid |\{S \in \mathcal A \mid S \ni u\}| \ge D\}$, i.e., the elements that are covered $\ge D$ times by $\mathcal A$.
\State Let $\mathcal B = \{S \in \mathcal A \mid |S \cap T| < 0.1\cdot S\}$, i.e., the sets in $\mathcal{A}$ that only contain few elements in $T$.
\State Allocate each $u \in U \setminus T$ to a set in $\mathcal B$ that contains it independently at random.
\State Let $\mathcal C = \{S \in \mathcal B \mid \text{$S$ collects $\ge \varepsilon\cdot |S \cap T|$ many elements}\}$.
\State Let $U'$ be the elements collected by sets in $\mathcal C$.
\State \textbf{Output}: $(\mathcal C,U')$.
\end{algorithmic}
\end{algorithm}

\begin{proof}[Proof of \Cref{thm:alg_sp}]
    Given the promise that there are $m$ disjoint sets in the \probsp{} instance, we have $\sum_{S} x_S \ge m$ where $\{x_S\}$ is an optimal solution to LP (\ref{eq:lp2}). Let $y_S$ be the 0-1 random variable indicating whether or not $S$ is picked into $\mathcal A$, where $y_S=1$ with probability $x_S$.

    Fix any element $u \in U$, we have
    $\mathbb E[\sum_{S \ni u} y_S]=\sum_{S \ni u} x_S \le 1$, i.e., in expectation it will be covered at most once. Consider the probability that it is covered at least $D$ times, for some large enough constant $D=D(\delta)$. By the Chernoff bound (\Cref{lem:chernoff}) ,
    $$
    \begin{aligned}
        \Pr\left[\sum_{S \ni u} y_S \ge D\right] \le & \exp\left(-\frac{(D-1)^2}{2+(D-1)}\right) & \quad{\text{(\Cref{lem:chernoff} with $\gamma=D-1$ and $\mu=1$)}} \\
        \le & \delta/20. & \quad{\text{(Since $D=40/\delta$)}}
    \end{aligned}$$
    Recall that $T=\{u \mid \sum_{S \ni u} y_S \ge D\}$. Thus for any fixed $u$, we have $\Pr[u \in T] \le \delta/20$.

    We bound the probability that a set is in $\mathcal A \setminus \mathcal B$, i.e., is discarded because of containing too many ``popular'' elements. For any fixed set $S$, we have
    $
        \mathbb E[|S \cap T|]  \le \delta/20 \cdot |S|$, by linearity of expectation and
        $\Pr[ |S \cap T|  \ge 0.1 \cdot |S|] \le \delta/2.$ by Markov's inequality. 
    Therefore,
    $$
    \begin{aligned}
        \mathbb E[| \mathcal B|] & = \sum_{S \in \mathcal S} \Pr[S \in \mathcal A] \cdot \Pr[S \in B \mid S \in \mathcal A] \\
        & = \sum_{S \in \mathcal S} \Pr[S \in \mathcal A]\cdot \Pr[|S \cap T| < 0.1 \cdot |S|\mid S \in \mathcal A]\\
        & \ge (1-\delta/2)\cdot \sum_{S \in \mathcal S} \Pr[S \in \mathcal A] \\
        & =(1-\delta/2)\cdot \mathbb E\left[|\mathcal A|\right],  
    \end{aligned}
    $$
    which means in expectation, 
    we do not lose too much by discarding the popular elements and only focusing on the sampled sets that contain mostly non-popular elements.

    From now on we only focus on the non-popular elements $U \setminus T$, i.e., the elements which are covered no more than $D$ times. We argue that by allocating each element $u \in U \setminus T$ to a random set $S\in \mathcal B$ that contains it, most sets will collect a reasonable fraction of their elements.

    For a fixed set $S\in \mathcal B$ and an element $u \in S \setminus T$,  since $u$ is covered $\le D$ times we have
    $$
    \begin{aligned}
        \Pr[u \text{ is allocated to } S] \ge 1/D.
    \end{aligned}
    $$
    Summing this over $S \setminus T$, we use linearity of expectation
    $$
    \begin{aligned}
        \mathbb E[\text{the number of elements $S$ collects}] \ge (1/D) \cdot |S \setminus T|.
    \end{aligned}
    $$
    Recall that our input sets have size $>C$. Fix a set $S \in \mathcal B$, we have $|S \setminus T| \ge 0.9 \cdot |S| \ge 0.9\cdot C$. Again using \Cref{lem:chernoff} with $\gamma=1-\varepsilon\cdot D$ and $\mu =(1/D) \cdot |S \setminus T| \ge 0.9\cdot C/D$, 
    \begin{align*}
         \Pr[\text{$S$ collects $<\varepsilon \cdot |S\setminus T|$ elements}]  \le & \exp\left(-\frac{(1-\varepsilon \cdot D)^2\cdot 0.9\cdot C/D}{2}\right)  
        \\
        \le & \exp\left(-\frac{0.5\cdot 0.9\cdot C/D}{2}\right) 
        \\
        < & \delta/2. 
    \end{align*}
    The second to last inequality follows due to the fact that  $(1-\varepsilon\cdot D)^2\ge 0.5$, and the last inequality since $C/D\ge 10/\delta$.
    Therefore,
    $$
    \begin{aligned}
    \mathbb E[|\mathcal C|] & = \sum_{S \in \mathcal S} \Pr[S \in \mathcal B] \cdot \Pr[S \in \mathcal C \mid S \in \mathcal B] \\  & = \sum_{S \in \mathcal S} \Pr[S \in \mathcal B] \cdot \Pr[\text{$S$ collects $\ge \varepsilon\cdot |S\setminus T|$ many elements} \mid S \in \mathcal B]\\
    & \ge (1-\delta/2)\cdot \sum_{S \in \mathcal S} \Pr[S \in \mathcal B] 
     = (1-\delta/2) \cdot \mathbb E[|\mathcal B|]\\
    & \ge (1-\delta/2)^2 \cdot \mathbb E[|\mathcal A|] 
     = (1-\delta/2)^2 \cdot \sum_{S \in \mathcal S} x_S 
     \ge (1-\delta) \cdot m,
    \end{aligned}
    $$
    which means in expectation, at least $(1-\delta) \cdot m$ many sets $S$ get at least $\varepsilon\cdot |S \setminus T|\ge 0.9\cdot \varepsilon\cdot |S|$ many elements.
\end{proof}

%% file: contents/hardness_sp.tex
\section{Hardness of $(1-\delta,1-\varepsilon)$-bicriteria \probsp{}}
\label{sec:hardness_sp}

In this section, we prove \Cref{thm:hardness_sp_2}. 

\hardsp*

To ease the notation, we refer to the property of sets $S_1,\ldots,S_{m'}$ in \Cref{thm:hardness_sp_2} as \emph{$\varepsilon$-almost disjoint}. Specifically, $S_1,\ldots,S_{m'}$ are $\varepsilon$-almost disjoint if and only if there exist sets $A_1\subseteq S_1,\ldots,A_{m'} \subseteq S_{m'}$, such that for every $i \in [m'], |A_i| \ge (1-\varepsilon) \cdot |S_i|$, and $A_1,\ldots,A_{m'}$ are disjoint.

We use the canonical reduction from $q$-\probcsp{} to \probsp{} \cite{HSS06, LST24}, while replacing the gadget in \cite{LST24} with Feige's hypercube partition system \cite{Fei98}, and plugging in Feige's hardness of degree-5 3-CNF \cite{Fei98}. We first give the definition of \probcsp{}, then show the reduction from \probcsp{} to \probsp{} and give the analysis.

\begin{definition}[CSP]
\label{def:csp}
    A $q$-\probcsp{} instance $\Pi=(G=(V,E),\Sigma,C=\{c_e\}_{e \in E})$ consists of $n=|V|$ variables and $m=|E|$ constraints, where each constraint $e\in E$ is on $q$ variables and is defined by the function $c_e:\Sigma^q \to \{0,1\}$. An assignment is a function from $V$ to $\Sigma$. The value of $\Pi$ is the maximum fraction of constraints that can be satisfied by some assignment.
\end{definition}

\paragraph{Reduction.}

Given a $d$-regular $q$-\probcsp{} instance over alphabet $\Sigma$, we construct a hypergraph gadget $G^{(x)}=(V^{(x)},E^{(x)})$ for each variable $x$ with the following properties:
\begin{enumerate}
    \item $E^{(x)}$ is partitioned into $|\Sigma|$ matchings $E^{(x)}_1,\ldots,E^{(x)}_{|\Sigma|}$ with $|E^{(x)}_1|=\ldots = |E^{(x)}_{|\Sigma|}| = d$, where $E^{(x)}_i = \{e^{(x)}_{i,j}\}_{j \in [d]}$.
    
    \item Any two hyperedges from two different matchings, i.e., any $e^{(x)}_{i_1,j_1}$ and $e^{(x)}_{i_2,j_2}$ for $i_1\neq i_2$ and $j_1,j_2 \in [d]$, must intersect in a reasonable fraction of their sizes.
\end{enumerate}

With this gadget, the only way to pick almost disjoint sets is to fix some $v \in \Sigma$, and take sets (a.k.a. hyperedges) from the matching $E^{(x)}_{v}$. This corresponds to the case where variable $x$ is assigned value $v$ in the \probcsp{} instance. For $j \in [d]$, the $j$-th hyperedge in $E^{(x)}_i$ corresponds to the $j$-th constraint involving $x$, assuming there is an arbitrary order of all constraints. In the final construction, we consider each constraint $\varphi$ on variables $(x_1,\ldots,x_q)$ and a satisfying assignment $(v_1,\ldots,v_q) \in \Sigma^q$, and build a set $S$ for them. Assuming $\varphi$ is the $j_i$-th constraint for variable $x_i$, the set $S$ is constructed to be
$$S=\bigcup_{i\in[q]} e^{(x_i)}_{v_i, j_i}.$$

Suppose the \probcsp{} instance has $m$ constraints and the average number of satisfying assignments for each constraints is $\eta$, then the total number of constructed sets is $\eta \cdot m$. 

\paragraph{Completeness.}
Suppose there is an assignment $\sigma$ that satisfies all constraints, we can for each constraint $\varphi$, take the set corresponding to $\sigma$ restricted on the variables involved in $\varphi$. By looking into each variable gadget, it's easy to see those sets are non-intersecting. Thus we can take $m$ non-intersecting sets out of the $\eta \cdot m$ sets.

\paragraph{Soundness.}
Suppose any assignment satisfies at most $\mu$ fraction of constraints. Then we prove one can take at most $\mu \cdot m$ many $\varepsilon$-almost non-intersecting sets, for $\varepsilon < \frac{\gamma}{2q}$ where $\gamma$ is the fraction of intersection of two hyperedges from two different matchings in the gadget.

Note that in order to be $\varepsilon$-almost non-intersecting, each set can drop at most an $\varepsilon$ fraction of elements. However, any two sets that correspond to inconsistent assignments share $\frac{\gamma}{q}>2\varepsilon$ fraction of elements. Therefore, the only way to pick $\varepsilon$-almost non-intersecting sets is to fix some global assignment $\sigma$ of the \probcsp{} instance, and pick sets that correspond to constraints satisfied by $\sigma$. Given $\sigma$ satisfies at most $\mu$ fraction of constraints, we can take at most $\mu \cdot m$ $\varepsilon$-almost non-intersecting sets in total.

\paragraph{Construction of the gadget.}
We use Feige's hypercube partition system \cite{Fei98} as the gadget. Specifically, the gadget $G^{(x)}=(V^{(x)},E^{(x)})$ is constructed as follows:

\begin{itemize}
    \item $V^{(x)} = [d]^{\Sigma}$. We can think of $V^{(x)}$ as $|\Sigma|$-long vectors, where each entry takes a value in $[d]$.
    \item $e^{(x)}_{i,j} = \{a \in [d]^{\Sigma} \mid a_{i} = j\}$, i.e., vectors whose $i$-th entry is $j$.
\end{itemize}

It is easy to see each $E^{(x)}_i = \{e^{(x)}_{i,j}\}_{j \in [d]}$ is a matching. Furthermore, $|e^{(x)}_{i,j}| = d^{|\Sigma|-1}$, and for $i_1\neq i_2$ and $j_1,j_2\in [d]$, $|e^{(x)}_{i_1,j_1} \cap e^{(x)}_{i_2,j_2}| = d^{|\Sigma|-2}$. In other words, two hyperedges from two different matchings intersect in $\frac{1}{d}$ of their sizes.

\medskip

\paragraph{Choice of Parameters.}

We use Feige's degree-5 3-SAT hardness \cite{Fei98} as a starting point.

\begin{lemma}[\cite{Fei98}]\label{lem:feige_3sat_hardness}
    There is a universal constant $\varepsilon_0>0$ such that it is NP-hard to distinguish between the case where a 3-CNF formula $\varphi$ is fully satisfiable, and the case where any assignment can satisfy at most $1-\varepsilon_0$ fraction of clauses. Furthermore, this holds even when each clause has length exactly 3 and each variable appears exactly 5 times.
\end{lemma}

Plugging everything into the reduction, we have:
\begin{itemize}
    \item $q=3$.
    \item $|\Sigma|=2$.
    \item $d=5$.
    \item $\eta = 7$, since each clause has exactly 7 satisfying assignments.
    \item $\gamma = \frac{1}{d} = \frac{1}{5}$.
    \item $\varepsilon <\frac{\gamma}{2q} = \frac{1}{30}$.
    \item $\alpha = \frac{1}{\eta} = \frac{1}{7}$.
    \item $\beta = \frac{1-\varepsilon_0}{\eta} = \frac{1}{7}(1-\varepsilon_0)$.
    \item $|\mathcal S| = \eta \cdot m = 7 \cdot m$.
    \item Every $S_i \in \mathcal S$ has size $k=q \cdot d^{|\Sigma|-1}= 15$.
    \item $|U|=k \cdot m=15 \cdot m$.
\end{itemize}

This finishes the proof of \Cref{thm:hardness_sp_2}.


%% file: contents/conclusion.tex
\section{Conclusions}\label{sec:conclude}

In this paper, we studied the bicriteria versions of several important scheduling problems, namely \probms{}, \probsc{}, and \probsp{}.

For \probms{}, we provided a bicriteria approximation algorithm: we can schedule more than 0.6533 fraction of the jobs within makespan $(3/2)\cdot T$, where $T$ is the optimal makespan to schedule all jobs.
We suggest as future work to prove hardness or algorithmic results
for scheduling a $(1-1/e+\varepsilon)$ fraction of jobs in makespan $(3/2-\delta)\cdot T$, for $\varepsilon>0.02$ and $\delta$ being a small constant.
Another trade-off point that would be interesting future work is scheduling a $(1-\varepsilon)$ fraction of jobs in makespan $(2-\delta)\cdot T$, for $\varepsilon,\delta$ small constants.

For \probsc{}, we provided hardness results with a constant fraction of agent rejections. An interesting open question would be closing the $\varepsilon$ vs. $(1-\varepsilon)$ gap between the algorithm in \cite{golovin2005max} and the hardness from us on the value each agent receives.

As a by-product, for \probsp{}, we provided the hardness for $(1-\delta_0,1-\varepsilon)$-bicriteria for some $\delta_0,\varepsilon>0$, and gave an algorithm for $(1-\delta,\varepsilon)$-bicriteria where $\delta$ can be arbitrarily close to 0. We conjecture $\varepsilon=1/2$ might be the threshold of going from easy to hard, since when $\varepsilon<1/2$, it's possible for a pair of sets to both pick an $\varepsilon$-fraction so that the subsets are disjoint, even if the two sets are identical. Thus, we suggest as future work to study the case when $\varepsilon$ is around $1/2$.

